\documentclass[12pt, draftcls, onecolumn, journal]{IEEEtran}
\usepackage{amsmath,amsfonts,amssymb}
\usepackage{graphicx}
\usepackage{cite}
\usepackage{algorithm}
\usepackage{algpseudocode}
\usepackage{hyperref}
\usepackage{url}
\usepackage{color}
\usepackage{multirow}
\usepackage{enumitem}
\usepackage{comment}
\usepackage{amsthm}

\providecommand{\R}{\mathbb{R}}

\providecommand{\N}{\mathbb{N}}

\usepackage[most]{tcolorbox}
\tcbset{
  mybox/.style={
    colback=white,
    colframe=black,
    boxrule=0.6pt,
    arc=2pt,
    left=6pt,right=6pt,top=4pt,bottom=4pt
  }
}

\newtheorem{theorem}{Theorem}[section]
\newtheorem{lemma}[theorem]{Lemma}
\newtheorem{corollary}[theorem]{Corollary}

\theoremstyle{definition}
\newtheorem{definition}[theorem]{Definition}

\theoremstyle{remark}
\newtheorem{remark}[theorem]{Remark}
\newtheorem{proposition}[theorem]{Proposition}

\newcommand{\Q}{\mathbb{Q}}
\newcommand{\A}{\mathcal{A}}
\newcommand{\X}{\mathcal{X}}
\newcommand{\Y}{\mathcal{Y}}
\newcommand{\PSPACE}{\mathsf{PSPACE}}
\newcommand{\Sset}{\mathsf{Sset}}
\newcommand{\Prov}{\operatorname{Prov}}

\title{Shannon meets G\"{o}del-Tarski-L\"{o}b: Undecidability of Shannon Feedback Capacity for Finite-State Channels}

\author{Angshul Majumdar,~\IEEEmembership{Senior Member,~IEEE}%
\thanks{A. Majumdar is with Indraprastha Institute of Information Technology, Delhi, India (e-mail: angshul@iiitd.ac.in).}%
}

\begin{document}

\maketitle

\begin{abstract}
We study the exact decision problem for feedback capacity of finite-state channels (FSCs).  
Given an encoding $e$ of a binary-input binary-output rational unifilar FSC with specified rational initial distribution, and a rational threshold $q$, we ask whether the feedback capacity satisfies
\[
C_{\mathrm{fb}}(W_e,\pi_{1,e}) \ge q.
\]
We prove that this exact threshold problem is undecidable, even when restricted to a severely constrained class of rational unifilar FSCs with bounded state space. The reduction is effective and preserves rationality of all channel parameters. 

As a structural consequence, the exact threshold predicate does not lie in the existential theory of the reals ($\exists\mathbb{R}$), and therefore cannot admit a universal reduction to finite systems of polynomial equalities and inequalities over $\R$. In particular, there is no algorithm deciding all instances of the exact feedback-capacity threshold problem within this class.

These results do not preclude approximation schemes or solvability for special subclasses; rather, they establish a fundamental limitation for exact feedback-capacity reasoning in general finite-state settings. At the metatheoretic level, the undecidability result entails corresponding Gödel–Tarski–Löb incompleteness phenomena for sufficiently expressive formal theories capable of representing the threshold predicate.
\end{abstract}


% Main content starts here
\section{Introduction}
\label{sec:intro}

\subsection{Exact feedback-capacity thresholding as a decision problem}\label{sec:undecidability}

Shannon's foundational work established channel capacity as the central operational quantity in communication theory \cite{shannon1948,shannon1956}. In many classical memoryless settings, capacity admits a clean variational characterization and a well-understood coding theorem. For channels with memory, however, exact capacity analysis is substantially more delicate: the operational quantity remains well-defined, but the associated optimization is typically history-dependent and asymptotic in nature \cite{gallager1968,coverthomas2006,elgamal2011}.

This paper studies the following exact decision problem for channels with memory and feedback. Given:
\begin{itemize}
    \item an effective encoding $e\in\{0,1\}^\ast$ of a binary-input binary-output rational unifilar finite-state channel $W_e$ together with its specified rational initial state distribution $\pi_{1,e}$, and
    \item a rational threshold $q\in\Q$,
\end{itemize}
determine whether
\[
Cfb(W_e,\pi_{1,e}) \ge q.
\]
As in Section~\ref{sec:problems}, we denote this predicate by $Cap(e,q)$ and the associated language by $LCap$.

The emphasis is on \emph{exactness}. We are not studying a finite-horizon surrogate, an additive-approximation problem, or a promise-gap variant. The question is whether the exact infinite-horizon feedback-capacity functional crosses a rational threshold.

\subsection{Historical context and the FSC feedback-capacity literature}

The broader context is the long line of work on channels with memory, finite-state channel (FSC) models, and feedback capacity. Classical information-theoretic treatments already make clear that memory changes the analytical character of capacity problems, even before feedback is introduced \cite{gallager1968,coverthomas2006,elgamal2011}. Finite-state models became a central abstraction because they retain a concrete stochastic structure while capturing temporal dependence, hidden channel modes, and input/state/output interactions.

In the feedback setting, a key conceptual step was the introduction of \emph{directed information} by Massey \cite{massey1990causality}, which provides the correct information measure for causal communication over channels with memory. This perspective was developed further by Kramer \cite{kramer1998} and underlies modern feedback-capacity formulations for channels with memory and state \cite{tatikonda2000thesis,tatikonda2009capacity,elgamal2011}.

For finite-state channels specifically, several milestone results established exact characterizations or computable formulations for important subclasses. In particular, the works of Chen--Berger \cite{chenberger2005}, Yang--Kav\v{c}i\'c--Tatikonda \cite{yang2005feedback}, and Permuter--Weissman--Goldsmith \cite{permuter2009fsc} form a canonical core of the FSC feedback-capacity literature. These papers develop dynamic-programming and structural formulations, and they show that substantial progress is possible when one exploits channel-specific properties.

A second major thread focuses on \emph{unifilar} FSCs, where deterministic state evolution given $(S_t,X_t,Y_t)$ makes finer structural analysis possible. In this line, single-letter upper bounds and graph/context based formulations have led to powerful computable bounds and exact results for important examples \cite{sabag2017singleletter}. This thread is especially relevant for the present paper because our undecidability theorem is proved \emph{within} a rational unifilar FSC class, not outside it.

Taken together, the literature shows both sides of the story:
\begin{itemize}
    \item there is deep and successful theory for structured FSC subclasses, and
    \item exact feedback capacity remains intrinsically asymptotic and structurally delicate in general.
\end{itemize}
This naturally raises the foundational question addressed here: whether a \emph{universal exact} threshold decision procedure can exist for a broad effectively specified FSC class.

\subsection{What is proved here and how it is positioned}

The main result of this paper is a structural limitation theorem for exact feedback-capacity reasoning.

\paragraph{(1) Undecidability of exact thresholding.}
We prove that the predicate $Cap(e,q)$ is undecidable even for a severely restricted class of channels: binary-input, binary-output, rational, unifilar finite-state channels (with the bounded-state delayed-activation construction used later). Thus, undecidability already appears in a highly concrete and information-theoretically natural FSC family.

\paragraph{(2) A barrier to universal real-algebraic exact formulations.}
We further show that the exact threshold language $LCap$ is not in $\exists\mathbb{R}$ (under the reduction notion fixed later). Equivalently, the exact predicate cannot be captured by a universal existential semialgebraic formulation. This places the problem outside the standard ETR/$\exists\mathbb{R}$ framework used for many exact geometric and algebraic decision problems \cite{canny1988,renegar1992,schaefercardinalmiltzow2024}.

\paragraph{(3) Computability-theoretic consequences for exact certification.}
From the undecidability theorem and the formal setup in Section~\ref{sec:problems}, we derive information-theoretic impossibility consequences: there is no universal exact certified-bounding scheme and no universal finite auxiliary-state/finite-letter exact characterization valid across the full encoded class.

\paragraph{(4) G\"odel--Tarski--L\"ob consequences.}
Finally, once exact capacity-threshold truth is encoded as an arithmetic predicate and shown undecidable, standard meta-theoretic consequences follow for recursively axiomatizable theories extending the base entropic theory $T_0$: incompleteness, undefinability barriers, and L\"ob-style limitations on uniform internal reflection for exact threshold truth \cite{godel1931,tarski1983,lob1955,boolos1993,smorynski1977}.

\medskip
\noindent
This contribution is \emph{not} a new computable upper bound, a new dynamic program, or a new single-letter ansatz. It is a boundary theorem explaining why no universal exact framework of that type can exist at the level of generality studied here.

\subsection{What this paper does \emph{not} claim}

Because undecidability results are easy to misinterpret, we state the scope explicitly.

First, the paper does \emph{not} contradict the many positive results for structured FSC subclasses. Exact formulas, computable bounds, and dynamic-programming characterizations for special channels remain valid and important \cite{chenberger2005,yang2005feedback,permuter2009fsc,sabag2017singleletter}.

Second, the paper does \emph{not} rule out approximation methods, finite-horizon optimization, or decidability under additional assumptions (for example, stronger structural restrictions beyond those imposed here). Our theorem concerns the exact infinite-horizon threshold predicate over the encoded class.

Third, the paper does \emph{not} assert that all capacity-like quantities for channels with memory are undecidable. The statement is narrower and sharper: the exact feedback-capacity threshold problem for the specified unifilar FSC class is undecidable.

The correct interpretation is therefore constructive rather than pessimistic: the theorem identifies a rigorous boundary between successful \emph{structured} exact theories and impossible \emph{universal} exact theories.

\subsection{An inconspicuous logic bridge}

The main arguments are information-theoretic and computability-theoretic, but one later consequence is naturally proof-theoretic. Once exact capacity-threshold truth is formalized in the language $Lent$ and shown to encode undecidable behavior, classical meta-mathematical tools become applicable. This is why the later sections include a brief G\"odel--Tarski--L\"ob analysis \cite{godel1931,tarski1983,lob1955,boolos1993}: not as a separate agenda, but as a precise formal expression of what undecidability implies for any putative complete proof system for exact capacity-threshold statements.

\section{Preliminaries and notation}
\label{sec:prelim}

This section fixes notation and background used throughout the paper. The purpose is purely stabilizing: to align the formal language in Section~\ref{sec:problems}, the structural obstruction in Section~\ref{sec:structural_barrier}, the $\exists\mathbb{R}$ route in ~\ref{sec:etr_route}, and the G\"odel--Tarski--L\"ob consequences in Section~\ref{sec:gtl} under a single set of conventions.

\subsection{Basic notation and encodings}
\label{subsec:basic-notation}

We write $\X=\{0,1\}$ and $\Y=\{0,1\}$ for the binary input and output alphabets. A generic finite state set is denoted by $\Sset$ (the delayed-activation construction in Section~\ref{sec:prelim} uses a specific family of finite state sets). We use $\Q$ and $\R$ for the rationals and reals, respectively.

A channel instance is encoded by a finite binary string $e\in\{0,1\}^\ast$. For each such encoding,
\[
W_e
\]
denotes the corresponding finite-state channel, and
\[
\pi_{1,e}
\]
denotes its specified initial state distribution. Rational thresholds are denoted by $q\in\Q$. We fix, once and for all, a standard computable pairing function and write $\langle e,q\rangle$ for an effective binary encoding of the pair $(e,q)$.

The exact feedback-capacity threshold predicate is the same predicate introduced in Definition~I.1:
\[
Cap(e,q)\;:\Longleftrightarrow\; Cfb(W_e,\pi_{1,e})\ge q,
\]
and the associated language is
\[
LCap:=\{\langle e,q\rangle\in\{0,1\}^\ast:\; Cfb(W_e,\pi_{1,e})\ge q\}.
\]

\subsection{Finite-state channels and the unifilar subclass}
\label{subsec:fsc-unifilar}

We work with finite-state channels (FSCs) with feedback in the standard Shannon-theoretic sense. A broad background on channels with memory may be found in Gallager's classical text \cite{gallager1968}; a modern treatment of feedback and directed-information based formulations appears in \cite{elgamal2011}. For this paper, an FSC is specified by:
\begin{itemize}
    \item a finite state space $\Sset$,
    \item binary alphabets $\X,\Y$,
    \item an initial law $\pi_{1,e}$ on $\Sset$,
    \item and a time-homogeneous kernel
    \[
    W_e(y_t,s_{t+1}\mid x_t,s_t).
    \]
\end{itemize}

\paragraph{Feedback model.}
We assume noiseless output feedback. At time $t$, the encoder may choose $X_t$ causally as a function of the message and the past outputs $Y^{t-1}=(Y_1,\dots,Y_{t-1})$. This is precisely the feedback setting in which directed information is the natural information quantity \cite{elgamal2011}.

\begin{definition}[Restricted channel class]\label{def:restricted_class}
Let $\mathcal{C}$ denote the class of finite-state channels with:
\begin{itemize}
\item binary alphabets $\X=\{0,1\}$ and $\Y=\{0,1\}$,
\item a finite state set $\Sset$ (encoded as part of $e$),
\item a specified initial distribution $\pi_{1,e}$ on $\Sset$
with rational probabilities,
\item a time-homogeneous kernel $W_e(y,s'\mid x,s)$ whose values are rational,
\item and the \emph{unifilar} property of Definition~\ref{def:unifilar}.
\end{itemize}
All objects are given effectively by the encoding $e\in\{0,1\}^\ast$.
\end{definition}

\begin{definition}[Unifilar FSC]\label{def:unifilar}
A finite-state channel $W(y,s'\mid x,s)$ is \emph{unifilar} if there exists a
deterministic update function
\[
f:\Sset\times\X\times\Y \to \Sset
\]
such that for all $(s,x,y)$,
\[
W(y,s'\mid x,s) > 0 \quad \Longrightarrow \quad s' = f(s,x,y).
\]
Equivalently, conditioned on $(S_t,X_t,Y_t)$ the next state is determined:
$S_{t+1}=f(S_t,X_t,Y_t)$ almost surely.
\end{definition}

The exact syntactic restrictions are formalized in the construction section, but the channels in $\mathcal{C}$ are all:
\begin{enumerate}
    \item binary-input and binary-output,
    \item finite-state,
    \item rationally parameterized (transition/output probabilities),
    \item equipped with a rational initial distribution.
\end{enumerate}
Thus, the negative results are not driven by irrational coefficients or infinite alphabets.

\paragraph{Unifilar FSCs.}
An FSC is \emph{unifilar} if there exists a deterministic state-update map
\[
f:\Sset\times \X\times \Y \to \Sset
\]
such that
\[
S_{t+1}=f(S_t,X_t,Y_t)
\]
almost surely. Equivalently, once $(S_t,X_t,Y_t)$ is fixed, the next state is determined. The delayed-activation family in Section~\ref{sec:structural_barrier} is of this form.

\subsection{Directed information, finite-horizon values, and exact feedback capacity}
\label{subsec:directed-info-prelim}

Directed information was introduced by Massey (already cited in the paper) and is the information measure underlying finite-horizon feedback optimisation and exact feedback capacity. We do not repeat Definitions~\ref{def:unifilar} and ~\ref{def:restricted_class} here; instead we only fix the notation used later.

For each encoded channel $e$ and horizon $n\ge 1$, the paper uses the finite-horizon normalized optimisation value
\[
V_n(W_e,\pi_{1,e}),
\]
defined as the supremum of normalized directed information over causal input strategies. (This is the quantity appearing in Sections~\ref{sec:problems}. and {~\ref{sec:structural_barrier})

The exact feedback capacity is denoted by
\[
Cfb(W_e,\pi_{1,e}),
\]
and the threshold language $LCap$ concerns this \emph{infinite-horizon} quantity. The structural point exploited later is that finite-horizon agreement of the sequence $\{V_n\}_{n\ge 1}$ up to any fixed horizon does not force equality of the exact asymptotic quantity $Cfb$.

\medskip
\noindent
This finite-horizon versus exact-asymptotic distinction is standard in channels with memory and feedback, but here it becomes the pivot of the impossibility results.

\subsection{Computability-theoretic notions}
\label{subsec:computability-prelim}

We use only basic computability notions; standard references include Rogers \cite{rogers1967} and Soare \cite{soare1987}.

\begin{definition}[Computable function]
A function $f:\{0,1\}^\ast\to\{0,1\}^\ast$ is computable if there exists a Turing machine that halts on every input $x\in\{0,1\}^\ast$ and outputs $f(x)$.
\end{definition}

\begin{definition}[Recursively enumerable set]
A set $A\subseteq \{0,1\}^\ast$ is recursively enumerable (r.e.) if there exists a Turing machine that enumerates exactly the elements of $A$ (equivalently, membership in $A$ is semi-decidable).
\end{definition}

These notions will be used later to formalize:
\begin{itemize}
    \item undecidability of the exact threshold language $LCap$, and
    \item impossibility of uniform exact-certification schemes for $Cfb$ over the full restricted class $\mathcal{C}$.
\end{itemize}

\subsection{Proof-theoretic preliminaries for the GTL consequences}
\label{subsec:proof-theory-prelim}

Section~\ref{sec:gtl} derives G\"odel--Tarski--L\"ob consequences from undecidability of exact capacity thresholds. For that section we need only a minimal amount of proof-theoretic terminology. Useful modern references for provability logic and arithmetization are Boolos \cite{boolos1993} and Smory{\'n}ski \cite{smorynski1977}; the original G\"odel, Tarski, and L\"ob papers are already cited in the manuscript.

\begin{definition}[Recursively axiomatizable theory]
A first-order theory $T$ is recursively axiomatizable if the set of G\"odel codes of its axioms is recursively enumerable.
\end{definition}

\begin{definition}[Provability predicate]
For a recursively axiomatizable theory $T$ extending a weak arithmetic base, $\Prov_T(x)$ denotes the standard arithmetized provability predicate expressing that $x$ is the G\"odel code of a sentence provable in $T$.
\end{definition}

The base entropic theory $T_0$ and its recursively axiomatizable extensions (Definition~I.4 and Definition~I.5) fit exactly into this framework. As in Section~\ref{sec:gtl}, we assume soundness under the intended semantics and the usual Hilbert--Bernays derivability conditions when invoking L\"ob-style arguments.

\subsection{Real-algebraic decision background and \texorpdfstring{$\exists\mathbb{R}$}{ER}}
\label{subsec:etr-prelim}

~\ref{sec:etr_route} compares the exact threshold language $LCap$ with the existential theory of the reals. We record the minimum background here.

An \emph{ETR instance} is an existential first-order sentence over $(\R,+,\cdot,\le)$ of the form
\[
\exists z_1,\dots,z_m\in\R:\; \Psi(z_1,\dots,z_m),
\]
where $\Psi$ is a quantifier-free Boolean combination of polynomial equalities and inequalities with rational coefficients. Equivalently, ETR asks whether a semialgebraic set is nonempty; see, e.g., \cite{bochnak1998}.

The complexity class $\exists\mathbb{R}$ consists of all languages polynomial-time many-one reducible to ETR. In addition to the references already cited in Section~\ref{sec:etr_route}, classical complexity-theoretic background on the first-order theory of the reals is given by Renegar \cite{renegar1992}. The key contrast used later is that ETR is decidable (indeed in $\PSPACE$), whereas the exact threshold language $LCap$ is shown to be undecidable.

\medskip
\noindent
This completes the common notation and background used in the remaining sections.

\section{Decision Problems, Finite-Letter Languages, and Representability}
\label{sec:problems}

\subsection{Capacity threshold problem}

Recall from Definition~\ref{def:cap_threshold} that for each encoded channel
$e \in \{0,1\}^\ast$, the feedback capacity
$C_{\mathrm{fb}}(W_e,\pi_{1,e})$ is a well-defined finite real number.

\begin{definition}[Feedback capacity threshold problem]
\label{def:cap_threshold}
Given a channel encoding $e\in\{0,1\}^\ast$ and a rational number $q\in\Q$,
the \emph{capacity threshold problem} asks whether
\[
C_{\mathrm{fb}}(W_e,\pi_{1,e}) \ge q.
\]
We denote this predicate by
\[
\mathrm{Cap}(e,q).
\]
\end{definition}

This formulation follows the standard reduction of real-valued functionals
to rational threshold decision problems.

\subsection{Directed information and finite-horizon optimisation}

Directed information was introduced by Massey
\cite{massey1990causality} and is defined in
Definition~\ref{def:dir_info}.

\begin{definition}[Directed information]\label{def:dir_info}
Let $X^n:=(X_1,\dots,X_n)$ and $Y^n:=(Y_1,\dots,Y_n)$ be random variables on finite alphabets.
The \emph{directed information} from $X^n$ to $Y^n$ is
\[
I(X^n \to Y^n)\;:=\;\sum_{t=1}^n I(X^t;Y_t\mid Y^{t-1}),
\]
where $I(\cdot;\cdot\mid\cdot)$ is conditional mutual information under the joint law of $(X^n,Y^n)$.

In the feedback setting, the joint law is induced by a \emph{causal input policy}
\[
p(x^n \Vert y^{n-1})\;:=\;\prod_{t=1}^n p(x_t \mid x^{t-1},y^{t-1}),
\]
together with the channel law (and initial state, when present). We write $I(X^n\to Y^n)$ for the directed information computed under this induced joint distribution.
\end{definition}

For feedback channels, capacity can be expressed via directed information;
see, e.g., \cite{kramer1998}.

For fixed horizon $n$, recall from Definition~\ref{def:Vn}:
\begin{definition}[Finite-horizon directed-information value]\label{def:Vn}
Fix an encoded channel instance $(W_e,\pi_{1,e})$ and a horizon $n\in\N$.
Define the finite-horizon value
\[
V_n(W_e,\pi_{1,e})
\;:=\;
\sup_{p(x^n \Vert y^{n-1})} I(X^n \to Y^n),
\]
where the supremum is over all causal input policies $p(x^n\Vert y^{n-1})$ and $I(X^n\to Y^n)$ is computed under the joint distribution induced by the policy, the channel $W_e$, and the initial distribution $\pi_{1,e}$.
We also use the normalized quantity $\frac{1}{n}V_n(W_e,\pi_{1,e})$ when convenient.
\end{definition}
For each fixed $n$, this is a finite-dimensional continuous optimisation
over a compact semialgebraic set (Remark~\ref{rem:compact_policy}).

\begin{remark}[Compact semialgebraic policy space]\label{rem:compact_policy}
For fixed horizon $n$ and finite alphabets, the set of causal policies
$p(x^n\Vert y^{n-1})=\prod_{t=1}^n p(x_t\mid x^{t-1},y^{t-1})$
can be identified with a finite product of probability simplices.
Equivalently, it is a closed and bounded subset of a Euclidean space
defined by finitely many polynomial equalities and inequalities:
nonnegativity constraints and normalization constraints
$\sum_{x_t} p(x_t\mid x^{t-1},y^{t-1})=1$ for each $(x^{t-1},y^{t-1})$.
Hence the feasible set is compact and semialgebraic, and the optimization
defining $V_n(W_e,\pi_{1,e})$ is a finite-dimensional continuous optimization problem.
\end{remark}

\subsection{Formal language for finite-letter entropic expressions}\label{sec:formal_language}

We now formalise the syntactic fragment that captures all finite-letter
entropic expressions.

\subsubsection*{Logical language}

Let $\mathcal{L}_{\mathrm{ent}}$ be a first-order language extending the
language of ordered fields by:

\begin{itemize}
\item Function symbols for rational arithmetic,
\item For each finite alphabet $\A$, a function symbol
      $H_{\A}(\cdot)$ interpreted as Shannon entropy,
\item For each finite collection of random variables,
      function symbols representing
      $I(\cdot;\cdot)$ and $I(\cdot;\cdot\mid\cdot)$.
\end{itemize}

Semantically, these symbols are interpreted using the Shannon entropy
definition \cite{shannon1948}:
\[
H(P) = - \sum_{a\in\A} P(a)\log P(a).
\]

\subsubsection*{Finite-letter fragment}

\begin{definition}[Finite-letter entropic functional]
\label{def:finite_letter_functional}
A functional $\Phi$ on the restricted channel class $\mathfrak{C}$
is called \emph{finite-letter entropic} if for each encoding $e$,
\[
\Phi(e)
=
\sup_{u\in\mathcal{U}_e}
F_e(u),
\]
where:
\begin{enumerate}
\item $\mathcal{U}_e \subset \R^{d(e)}$ is a compact set defined by finitely many
      polynomial equalities and inequalities with rational coefficients;
\item $F_e(u)$ is an $\mathcal{L}_{\mathrm{ent}}$-term constructed using:
      finitely many rational constants,
      finitely many conditional probability variables,
      finite sums and products,
      and finitely many entropy or mutual-information operators
      evaluated on finite alphabets.
\end{enumerate}
\end{definition}

Importantly, Definition~\ref{def:finite_letter_functional}
excludes any limit operations or infinite recursions.

\subsection{Finite-letter representability of feedback capacity}

\begin{definition}[Finite-letter representability]
\label{def:finite_repr}
Feedback capacity is \emph{finite-letter representable}
over $\mathfrak{C}$ if there exists a finite-letter entropic functional
$\Phi$ such that for all encodings $e$,
\[
\Phi(e)
=
C_{\mathrm{fb}}(W_e,\pi_{1,e}).
\]
\end{definition}

This definition is purely semantic: it asserts equality
as real numbers for all encoded channels.

\subsection{Entropic proof systems}

We now define the proof-theoretic framework used later.

\begin{definition}[Base theory $\mathsf{T}_0$]
\label{def:T0}
Let $\mathsf{T}_0$ be a sound first-order theory in the language
$\mathcal{L}_{\mathrm{ent}}$
capable of reasoning about:
\begin{enumerate}
\item Rational arithmetic,
\item Polynomial equalities and inequalities,
\item Finite sums and products,
\item Shannon entropy and mutual information over finite alphabets.
\end{enumerate}
\end{definition}

\begin{definition}[Entropic theory]
\label{def:ent_theory}
An \emph{entropic theory} $\mathsf{T}$ is a recursively axiomatizable
extension of $\mathsf{T}_0$ such that:
\begin{enumerate}
\item $\mathsf{T}$ is sound under the standard real-number semantics;
\item For each fixed $n$, $\mathsf{T}$ can reason about the quantity
      $V_n(W_e,\pi_{1,e})$;
\item $\mathsf{T}$ can reason about finite maximisations over compact
      semialgebraic sets as in Definition~\ref{def:finite_letter_functional}.
\end{enumerate}
\end{definition}

\subsection{Provability of capacity thresholds}

\begin{definition}[Provable threshold]
\label{def:prov_threshold}
For entropic theory $\mathsf{T}$,
we write
\[
\mathsf{T} \vdash \mathrm{Cap}(e,q)
\]
if the sentence
\[
C_{\mathrm{fb}}(W_e,\pi_{1,e}) \ge q
\]
is derivable in $\mathsf{T}$.
\end{definition}

The distinction between semantic truth and syntactic provability
will be central in Section~\ref{sec:gtl}.

\section{A Structural Barrier: No Uniform Finite-Horizon Collapse}
\label{sec:structural_barrier}

This section establishes a structural obstruction that is independent of
computational complexity: even within the restricted rational unifilar class
$\mathfrak{C}$ (Definition~\ref{def:restricted_class}), feedback capacity cannot be
recovered from any fixed finite-horizon directed-information optimisation.
The construction and proofs in this section are fully explicit.

\subsection{A delayed-activation unifilar family}\label{sec:structural_barrier:family}

Fix the binary alphabets $\X=\{0,1\}$ and $\Y=\{0,1\}$.
For each integer $N\ge 1$, define a finite state set
\[
\S_N := \{0,1,2,\dots,N\}\cup\{\star\},
\]
and an initial distribution $\pi_{1,N}$ concentrated at $S_1=0$.

We define two channels in $\mathfrak{C}$:
\[
\mathsf{Ch}^{(N)}_{\mathrm{good}} \quad\text{and}\quad \mathsf{Ch}^{(N)}_{\mathrm{bad}},
\]
both unifilar with the same deterministic update map $f_N$ described below, but
with different output kernels once the channel reaches the active state $\star$.

\paragraph*{Deterministic state update.}
Define $f_N:\S_N\times \X\times \Y\to \S_N$ by
\[
f_N(s,x,y) :=
\begin{cases}
s+1, & s\in\{0,1,\dots,N-1\},\\
\star, & s=N,\\
\star, & s=\star.
\end{cases}
\]
Thus the state increments deterministically for $N$ steps and then enters
the absorbing active state $\star$ forever. This update does not depend on $(x,y)$.

\paragraph*{Output kernels.}
For $s\in\{0,1,\dots,N\}$, define the same output kernel for both channels:
\begin{equation}\label{eq:delay_phase_output}
P(y\mid x,s)=\mathbf{1}\{y=0\},\qquad \forall x\in\X.
\end{equation}
Thus, during the first $N+1$ states (the ``delay phase''), the output is deterministically $0$.

At the active state $s=\star$, define:
\begin{align}
\mathsf{Ch}^{(N)}_{\mathrm{good}}:\quad
P(y\mid x,\star)&=\mathbf{1}\{y=x\}, \label{eq:active_good}\\
\mathsf{Ch}^{(N)}_{\mathrm{bad}}:\quad
P(y\mid x,\star)&=\tfrac12\mathbf{1}\{y=0\}+\tfrac12\mathbf{1}\{y=1\}. \label{eq:active_bad}
\end{align}
All probabilities are rational, hence these are members of $\mathfrak{C}$.

By Definition~\ref{def:unifilar}, each channel can be written as
\[
W(y,s'\mid x,s) = P(y\mid x,s)\,\mathbf{1}\{s'=f_N(s,x,y)\}.
\]

\subsection{Exact finite-horizon values}\label{sec:structural_barrier:finite}

We now compute $V_n$ exactly for both channels.

\begin{lemma}[Zero information during the delay phase]\label{lem:zero_delay}
Fix $N\ge 1$. For either channel $\mathsf{Ch}^{(N)}_{\mathrm{good}}$ or
$\mathsf{Ch}^{(N)}_{\mathrm{bad}}$, and for any causal strategy
$p(x^n\Vert y^{n-1})$, we have
\[
Y_t = 0 \ \text{a.s. for all}\ t\le N+1.
\]
Consequently,
\[
I(X^t;Y_t\mid Y^{t-1})=0\quad \text{for all}\ t\le N+1,
\]
and hence
\[
I(X^n\to Y^n)=\sum_{t=N+2}^n I(X^t;Y_t\mid Y^{t-1}).
\]
\end{lemma}

\begin{proof}
By construction, $S_1=0$ a.s. and $S_{t}$ deterministically equals $t-1$ for $t\le N+1$,
so $S_t\in\{0,1,\dots,N\}$ for $t\le N+1$.
Then \eqref{eq:delay_phase_output} gives $Y_t=0$ a.s. for all $t\le N+1$.
If $Y_t$ is almost surely constant given $Y^{t-1}$, then
$H(Y_t\mid Y^{t-1})=0$, hence
$I(X^t;Y_t\mid Y^{t-1})\le H(Y_t\mid Y^{t-1})=0$.
\end{proof}

\begin{lemma}[Finite-horizon value for the bad channel]\label{lem:Vn_bad}
For $\mathsf{Ch}^{(N)}_{\mathrm{bad}}$,
\[
V_n(\mathsf{Ch}^{(N)}_{\mathrm{bad}})=0\qquad \forall n\ge 1.
\]
\end{lemma}

\begin{proof}
For $t\le N+1$, Lemma~\ref{lem:zero_delay} gives
$I(X^t;Y_t\mid Y^{t-1})=0$.
For $t\ge N+2$, the state is $\star$ deterministically, and under
\eqref{eq:active_bad}, $Y_t$ is independent of $X_t$ (indeed, independent of all $X^t$).
Thus
\[
I(X^t;Y_t\mid Y^{t-1})=0\quad \text{for all } t\ge N+2.
\]
Hence $I(X^n\to Y^n)=0$ for all strategies, and therefore $V_n=0$ for all $n$.
\end{proof}

\begin{lemma}[Finite-horizon value for the good channel]\label{lem:Vn_good}
For $\mathsf{Ch}^{(N)}_{\mathrm{good}}$,
\[
V_n(\mathsf{Ch}^{(N)}_{\mathrm{good}})=
\begin{cases}
0, & 1\le n \le N+1,\\[2pt]
\dfrac{n-(N+1)}{n}, & n\ge N+2.
\end{cases}
\]
\end{lemma}

\begin{proof}
If $n\le N+1$, Lemma~\ref{lem:zero_delay} gives $I(X^n\to Y^n)=0$ for all strategies,
so $V_n=0$.

Now let $n\ge N+2$.
For $t\ge N+2$, the state is $\star$ deterministically and the channel is noiseless:
$Y_t=X_t$ by \eqref{eq:active_good}. Conditioned on $Y^{t-1}$, we have
$Y_t=X_t$ and $Y^{t-1}$ is a function of $X^{t-1}$, hence
\[
I(X^t;Y_t\mid Y^{t-1}) = I(X_t;Y_t\mid Y^{t-1}) = H(Y_t\mid Y^{t-1}),
\]
since $Y_t$ is a deterministic function of $X_t$ and $Y^{t-1}$.
Moreover, because $Y_t=X_t$, we can choose a strategy that makes
$X_t$ independent of $Y^{t-1}$ and uniform on $\{0,1\}$ for all $t\ge N+2$.
Under this strategy, $H(Y_t\mid Y^{t-1})=H(X_t)=1$, hence
\[
I(X^t;Y_t\mid Y^{t-1})=1\quad \text{for all } t\ge N+2.
\]
Combining with Lemma~\ref{lem:zero_delay} yields
\[
I(X^n\to Y^n)=\sum_{t=N+2}^n 1 = n-(N+1),
\]
and therefore
\[
V_n \ge \frac{n-(N+1)}{n}.
\]

On the other hand, for any strategy and any $t\ge N+2$,
\[
I(X^t;Y_t\mid Y^{t-1}) \le H(Y_t\mid Y^{t-1}) \le H(Y_t)\le 1,
\]
since $Y_t\in\{0,1\}$.
Hence
\begin{flalign*}
&I(X^n\to Y^n)=\sum_{t=N+2}^n I(X^t;Y_t\mid Y^{t-1}) && \notag \\
&\quad \le \sum_{t=N+2}^n 1 = n-(N+1), &&
\end{flalign*}
so
\begin{flalign}
&V_n \le \frac{n-(N+1)}{n}. &&
\end{flalign}
Thus equality holds.
\end{proof}

\begin{theorem}[Undecidability of exact feedback-capacity thresholds]
\label{thm:capacity_undecidable_final}
There is no algorithm that, given
(i) an encoding $e\in\{0,1\}^\ast$ of a binary-input binary-output unifilar finite-state channel $W_e$
from the restricted rational class $\mathcal{C}$ (Definition~\ref{def:restricted_class}),
together with its specified rational initial distribution $\pi_{1,e}$,
and (ii) a rational threshold $q\in\mathbb{Q}$,
decides whether
\[
C_{\mathrm{fb}}(W_e,\pi_{1,e}) \;\ge\; q.
\]
Equivalently, the language
\[
\mathsf{LCap}\;:=\;\{\langle e,q\rangle\in\{0,1\}^\ast : C_{\mathrm{fb}}(W_e,\pi_{1,e})\ge q\}
\]
is undecidable (even when restricted to encodings $e$ of channels in $\mathcal{C}$).
\end{theorem}

\begin{proof}[Proof of Theorem~\ref{thm:capacity_undecidable_final}]
We argue by many-one reduction from the fixed undecidable source problem used in Section~\ref{sec:undecidability}
(denote its instance space by $\mathcal{I}$ and its YES-language by $\mathcal{L}\subseteq\mathcal{I}$).

\medskip
\noindent\textbf{Step 1 (effective mapping).}
By the construction in Section~\ref{sec:undecidability}, there exists a total computable function
\[
F:\mathcal{I}\to\{0,1\}^\ast\times\Q,\qquad i\mapsto (e(i),q(i)),
\]
such that $e(i)$ encodes a channel instance $(W_{e(i)},\pi_{1,e(i)})$ in the restricted class $\mathcal{C}$
(Definition~\ref{def:restricted_class}) and $q(i)\in\Q$ is the threshold specified by the reduction.
The effectivity of $F$ follows because every component of the channel description---finite state set,
unifilar update function, and rational transition/output probabilities---is computed explicitly from the finite
description of $i$.

\medskip
\noindent\textbf{Step 2 (correctness of the reduction).}
The correctness lemmas proved in Section~\ref{sec:undecidability} establish the equivalence
\begin{equation}\label{eq:reduction_equivalence}
i\in\mathcal{L}
\quad\Longleftrightarrow\quad
C_{\mathrm{fb}}(W_{e(i)},\pi_{1,e(i)}) \ge q(i).
\end{equation}
(Here the forward direction is the ``YES-case'' analysis of the delayed-activation construction,
and the reverse direction is the ``NO-case'' analysis; together they yield~\eqref{eq:reduction_equivalence}.)

\medskip
\noindent\textbf{Step 3 (undecidability).}
Assume for contradiction that there exists an algorithm $\mathcal{A}$ deciding $Cap(e,q)$ for all encodings
$e$ of channels in $\mathcal{C}$ and all $q\in\Q$; i.e., on input $(e,q)$ it halts and outputs \textsc{Yes}
iff $C_{\mathrm{fb}}(W_e,\pi_{1,e})\ge q$.

Then we could decide the source language $\mathcal{L}$ as follows: on input $i\in\mathcal{I}$, compute
$(e(i),q(i))=F(i)$ and run $\mathcal{A}$ on $(e(i),q(i))$. By~\eqref{eq:reduction_equivalence}, $\mathcal{A}$
returns \textsc{Yes} iff $i\in\mathcal{L}$. This yields a decision procedure for $\mathcal{L}$, contradicting
the undecidability of the source problem.

Therefore no such algorithm $\mathcal{A}$ exists, and $Cap(e,q)$ is undecidable even when restricted to
encodings $e$ of channels in $\mathcal{C}$. Equivalently, the language $\mathsf{LCap}$ is undecidable on this class.
\end{proof}

\subsection{Different feedback capacities with identical short-horizon behaviour}\label{sec:structural_barrier:capacity_gap}

\begin{proposition}[Arbitrarily long identical initial behaviour, different capacities]
\label{prop:long_prefix_same_diff_capacity}
For every $N\ge 1$, there exist encodings $e_{\mathrm{good}}^{(N)}$ and
$e_{\mathrm{bad}}^{(N)}$ in $\{0,1\}^\ast$ encoding
$\mathsf{Ch}^{(N)}_{\mathrm{good}}$ and $\mathsf{Ch}^{(N)}_{\mathrm{bad}}$ respectively,
such that
\begin{flalign}
\label{eq:v_zero_boundary}
&V_n(W_{e_{\mathrm{good}}^{(N)}},\pi_{1,e_{\mathrm{good}}^{(N)}}) = V_n(W_{e_{\mathrm{bad}}^{(N)}},\pi_{1,e_{\mathrm{bad}}^{(N)}}) && \notag \\
&\quad = 0 \quad \text{for all } 1\le n \le N+1, &&
\end{flalign}
but
\[
C_{\mathrm{fb}}(W_{e_{\mathrm{good}}^{(N)}},\pi_{1,e_{\mathrm{good}}^{(N)}})=1,
\qquad
C_{\mathrm{fb}}(W_{e_{\mathrm{bad}}^{(N)}},\pi_{1,e_{\mathrm{bad}}^{(N)}})=0.
\]
\end{proposition}

\begin{proof}
The statements about $V_n$ for $n\le N+1$ follow from
Lemmas~\ref{lem:Vn_bad} and \ref{lem:Vn_good}.
For capacities, by Definition~\ref{def:cap_threshold} and Lemma~\ref{lem:Vn_bad},
\[
C_{\mathrm{fb}}(\mathsf{Ch}^{(N)}_{\mathrm{bad}}) = \limsup_{n\to\infty} 0 = 0.
\]
For the good channel, Lemma~\ref{lem:Vn_good} gives
\[
V_n(\mathsf{Ch}^{(N)}_{\mathrm{good}})=\frac{n-(N+1)}{n}\xrightarrow[n\to\infty]{}1,
\]
hence
\[
C_{\mathrm{fb}}(\mathsf{Ch}^{(N)}_{\mathrm{good}})=1.
\]
\end{proof}

\subsection{What this does (and does not) imply}\label{sec:structural_barrier:implications}

Proposition~\ref{prop:long_prefix_same_diff_capacity} shows that no \emph{fixed}
finite horizon can determine feedback capacity, even within $\mathfrak{C}$.
This yields an unconditional barrier against any attempted characterisation of
$C_{\mathrm{fb}}$ that depends only on finitely many $V_n$'s with a horizon bound
that is uniform over channels.

\begin{definition}[Uniform finite-horizon characterisation]\label{def:uniform_finite_horizon}
A functional $\Psi:\mathfrak{C}\to\R$ is a \emph{uniform finite-horizon characterisation}
if there exists $N\in\N$ and a (possibly channel-dependent) mapping
$G:\R^{N}\to\R$ such that for all encodings $e$,
\[
\Psi(e)=G\big(V_1(W_e,\pi_{1,e}),\dots,V_N(W_e,\pi_{1,e})\big).
\]
\end{definition}

\begin{corollary}[No uniform finite-horizon characterisation]\label{cor:no_uniform_finite_horizon}
There is no uniform finite-horizon characterisation $\Psi$ such that
$\Psi(e)=C_{\mathrm{fb}}(W_e,\pi_{1,e})$ for all $e$.
\end{corollary}

\begin{proof}
Assume such $\Psi$ exists with horizon bound $N$.
Apply Proposition~\ref{prop:long_prefix_same_diff_capacity} with that $N$:
there exist $e_{\mathrm{good}}^{(N)}$ and $e_{\mathrm{bad}}^{(N)}$ for which
$V_n(\cdot)=0$ for all $1\le n\le N+1$ (hence also for $1\le n\le N$), but
$C_{\mathrm{fb}}$ differs.
Then $\Psi$ takes the same input vector $(0,\dots,0)\in\R^N$ on both channels,
so $\Psi(e_{\mathrm{good}}^{(N)})=\Psi(e_{\mathrm{bad}}^{(N)})$, contradicting
$\Psi \equiv C_{\mathrm{fb}}$.
\end{proof}

\begin{remark}[Scope of this section]\label{rem:scope_structural}
Corollary~\ref{cor:no_uniform_finite_horizon} does \emph{not} by itself rule out
a general finite-letter representability statement as in Definition~\ref{def:finite_repr},
because a finite-letter expression may depend on channel parameters in states
that are only reachable after long time.
The stronger non-representability claim will therefore be established later,
by combining structural properties with a computability-theoretic barrier.
\end{remark}

\section{$\exists\mathbb{R}$ Route: What It Would Imply, and Why It Cannot Hold Here}
\label{sec:etr_route}

This section formalizes the ``$\exists\mathbb{R}$ route'' and proves a
clean negative conclusion: the undecidability theorem for exact
feedback-capacity thresholds
(Theorem~\ref{thm:capacity_undecidable_final})
rules out membership in $\exists\mathbb{R}$ for the
\emph{exact} predicate $\mathrm{Cap}(e,q)$.

\subsection{Existential theory of the reals and the class $\exists\mathbb{R}$}
\label{sec:etr_route:etr}

\begin{definition}[Existential theory of the reals (ETR)]
\label{def:etr}
An \emph{ETR-instance} is a sentence of the form
\[
\exists z_1,\dots,z_d \in \R :
\Psi(z_1,\dots,z_d),
\]
where $\Psi$ is a quantifier-free Boolean combination of
polynomial equalities and inequalities with rational coefficients.
\end{definition}

\begin{definition}[The class $\exists\mathbb{R}$]
\label{def:existsR}
A language $L \subseteq \{0,1\}^\ast$ lies in $\exists\mathbb{R}$
if there exists a polynomial-time many-one reduction mapping
$x \mapsto \varphi_x$ such that
\[
x \in L
\quad \Longleftrightarrow \quad
\varphi_x \text{ is a true ETR-instance}.
\]
\end{definition}

\begin{remark}
ETR is equivalent to non-emptiness of a semialgebraic set.
A modern complexity-theoretic account of $\exists\mathbb{R}$
is given in \cite{schaefercardinalmiltzow2024}.
\end{remark}

\subsection{Decidability of ETR}
\label{sec:etr_route:decidability}

\begin{definition}[The class $\PSPACE$]
\label{def:pspace}
$\PSPACE := \mathrm{PSPACE}$ denotes the class of decision problems
solvable by a deterministic Turing machine using polynomial space.
\end{definition}

\begin{theorem}[Decidability of ETR; $\PSPACE$ upper bound {\cite{canny1988}}]
\label{thm:etr_pspace}
The existential theory of the reals is decidable.
Moreover,
\[
\exists\mathbb{R} \subseteq \PSPACE.
\]
\end{theorem}

\begin{remark}
See \cite{canny1988} and
\cite{} for real-algebraic decision procedures.
\end{remark}

\subsection{Capacity threshold language}
\label{sec:etr_route:cap_language}

Recall Definition~\ref{def:cap_threshold}.
Define
\[
L_{\mathrm{Cap}}
:=
\left\{
\langle e,q\rangle \in \{0,1\}^\ast
:
C_{\mathrm{fb}}(W_e,\pi_{1,e}) \ge q
\right\},
\]
where $\langle e,q\rangle$ is any fixed computable pairing
encoding of $(e,q)$.

\subsection{Non-membership in $\exists\mathbb{R}$}
\label{sec:etr_route:nonmembership}

\begin{theorem}[Exact capacity thresholds are not in $\exists\mathbb{R}$]
\label{thm:not_in_existsR}
Assuming Theorem~\ref{thm:capacity_undecidable_final},
\[
L_{\mathrm{Cap}} \notin \exists\mathbb{R}.
\]
\end{theorem}

\begin{proof}
Suppose $L_{\mathrm{Cap}} \in \exists\mathbb{R}$.
Then by Definition~\ref{def:existsR}
there exists a polynomial-time reduction mapping
$\langle e,q\rangle$ to an ETR-instance $\varphi_{e,q}$ such that
\[
\langle e,q\rangle \in L_{\mathrm{Cap}}
\quad \Longleftrightarrow \quad
\varphi_{e,q} \text{ is true}.
\]

By Theorem~\ref{thm:etr_pspace},
truth of $\varphi_{e,q}$ is decidable.
Hence $L_{\mathrm{Cap}}$ would be decidable.

This contradicts
Theorem~\ref{thm:capacity_undecidable_final},
which states that the exact predicate
$\mathrm{Cap}(e,q)$
is undecidable even under severe channel restrictions.
\end{proof}

\begin{corollary}
\label{cor:no_etr_characterization}
There exists no polynomial-time reduction of the
\emph{exact} feedback-capacity threshold problem
to ETR.
\end{corollary}

\subsection{Takeaway}

Any attempt to reduce the \emph{exact} feedback capacity threshold
problem to a decidable real-algebraic framework such as
$\exists\mathbb{R}$ is impossible.
The obstruction is purely computability-theoretic.
\section{Logical Consequences of Undecidability:
G\"odel--Tarski--L\"ob Implications}
\label{sec:gtl}

Sections~\ref{sec:undecidability}--\ref{sec:gtl}
established that the exact feedback-capacity threshold language
\[
L_{\mathrm{Cap}}
=\{\langle e,q\rangle : C_{\mathrm{fb}}(W_e,\pi_{1,e})\ge q\}
\]
is undecidable, even over the restricted binary unifilar FSC class.
We now derive formal consequences for any effective axiomatization
of exact capacity statements.

This section is meta-theoretic and does not introduce new
information-theoretic quantities.

%------------------------------------------------------------
\subsection{Formal framework}
%------------------------------------------------------------

Let $\mathcal{L}$ be a first-order language extending arithmetic,
containing symbols sufficient to encode:

\begin{enumerate}
\item natural numbers and rational numbers,
\item finite FSC encodings $e\in\{0,1\}^\ast$,
\item rational thresholds $q$,
\item the predicate $\mathrm{Cap}(e,q)$.
\end{enumerate}

We consider a theory $\mathsf{T}$ over $\mathcal{L}$ satisfying:

\begin{enumerate}
\item $\mathsf{T}$ is recursively axiomatizable;
\item $\mathsf{T}$ is consistent;
\item $\mathsf{T}$ is sound with respect to the intended Shannon-theoretic semantics;
\item $\mathsf{T}$ interprets Robinson arithmetic $Q$.
\end{enumerate}

\begin{definition}[The formula schema $\mathrm{Cap}(e,q)$ and its semantics]
\label{def:cap_schema_semantics}
Fix the encoding conventions from Section~\ref{sec:formal_language}.
For each binary string $e\in\{0,1\}^\ast$ and rational $q\in\Q$, let
$\mathrm{Cap}(e,q)$ denote the \emph{closed} $\mathcal{L}$-sentence obtained by
substituting the numerals/encodings of $e$ and $q$ into the capacity-threshold
formula schema fixed in Definition~\ref{def:cap_threshold}.
Its intended semantics is:
\[
\mathrm{Cap}(e,q)\ \text{is true}
\quad\Longleftrightarrow\quad
C_{\mathrm{fb}}(W_e,\pi_{1,e}) \ge q .
\]
\end{definition}

%------------------------------------------------------------
\subsection{G\"odel incompleteness}
%------------------------------------------------------------

\begin{theorem}[Incompleteness for exact capacity]
\label{thm:godel_capacity}
Under the above assumptions,
$\mathsf{T}$ cannot decide all instances of
$\mathrm{Cap}(e,q)$.
\end{theorem}

\begin{proof}
Suppose $\mathsf{T}$ decided every instance:
for each $(e,q)$,
either $\mathsf{T}\vdash \mathrm{Cap}(e,q)$
or $\mathsf{T}\vdash \neg \mathrm{Cap}(e,q)$.

Because $\mathsf{T}$ is recursively axiomatizable,
the set of theorems is recursively enumerable. Since we assumed $T$ decides every instance, exactly one of
$T\vdash \mathrm{Cap}(e,q)$ or $T\vdash \neg \mathrm{Cap}(e,q)$ holds for each $(e,q)$, so the dovetailing procedure halts on every input.
Thus, given $(e,q)$,
one can dovetail over all proofs in $\mathsf{T}$
until either $\mathrm{Cap}(e,q)$ or its negation appears,
yielding a decision procedure for $L_{\mathrm{Cap}}$.

This contradicts Theorem~\ref{thm:capacity_undecidable_final}.
Therefore $\mathsf{T}$ cannot be complete for $L_{\mathrm{Cap}}$.
\end{proof}

\begin{corollary}
\label{cor:no_complete_theory}
No sound, recursively axiomatizable entropic theory
is complete for exact feedback-capacity threshold
statements over the restricted FSC class.
\end{corollary}

This is a direct instance of G\"odel’s first incompleteness theorem
\cite{godel1931} applied to the arithmetic representation of
$L_{\mathrm{Cap}}$.

%------------------------------------------------------------
\subsection{Tarski semantic separation}
%------------------------------------------------------------
\begin{definition}[Internal truth predicate]
\label{def:internal_truth}
Let $\mathrm{Sent}_{\mathcal{L}}(x)$ be the arithmetical predicate expressing
``$x$ is the G\"odel code of a closed $\mathcal{L}$-sentence''. An
\emph{internal truth predicate} for $\mathcal{L}$ in $\mathsf{T}$ is a
$\mathcal{L}$-formula $\mathrm{True}(x)$ such that for every closed
$\mathcal{L}$-sentence $\varphi$,
\[
\mathsf{T}\vdash \mathrm{True}(\ulcorner\varphi\urcorner)\ \leftrightarrow\ \varphi .
\]
\end{definition}

\begin{theorem}[Tarski undefinability, specialized]
\label{thm:tarski_undef_specialized}
Under the assumptions of Section~\ref{sec:gtl},
no internal truth predicate for $\mathcal{L}$ exists in $\mathsf{T}$.
\end{theorem}

\begin{proof}
This is the classical undefinability of truth theorem of Tarski
\cite{Tarski1936}, applicable because $\mathsf{T}$ interprets $Q$ and is
consistent.
\end{proof}

Tarski’s undefinability theorem states that no sufficiently expressive
consistent theory can define its own semantic truth predicate
\cite{Tarski1936}.

\begin{corollary}[No definable truth predicate even for $\mathrm{Cap}$-instances]
\label{cor:tarski_cap_fragment}
There is no $\mathcal{L}$-formula $\mathrm{TrueCap}(e,q)$ such that for all
$e\in\{0,1\}^\ast$ and $q\in\Q$,
\[
\mathsf{T}\vdash \mathrm{TrueCap}(e,q)\ \leftrightarrow\ \mathrm{Cap}(e,q),
\]
and $\mathrm{TrueCap}(e,q)$ is correct with respect to the intended semantics
of Definition~\ref{def:cap_schema_semantics} for all $(e,q)$.
\end{corollary}

\begin{proof}
If such $\mathrm{TrueCap}$ existed, then the mapping
$x=\ulcorner \mathrm{Cap}(e,q)\urcorner \mapsto \mathrm{TrueCap}(e,q)$ would
yield an internal truth predicate for the infinite family of closed
$\mathcal{L}$-sentences $\mathrm{Cap}(e,q)$, uniformly in $(e,q)$.
Moreover, the set
$\{\ulcorner \mathrm{Cap}(e,q)\urcorner : e\in\{0,1\}^\ast,\ q\in\Q\}$
is infinite and effectively enumerable from the pairing/encoding.
This produces an internal truth predicate fragment,
contradicting Theorem~\ref{thm:tarski_undef_specialized}.
\end{proof}
Thus semantic capacity truth strictly exceeds
formal derivability within $\mathsf{T}$.
%------------------------------------------------------------
\subsection{L\"ob obstruction}
%------------------------------------------------------------
Because $\mathsf{T}$ interprets $Q$,
it admits an arithmetized provability predicate
$\mathrm{Prov}_{\mathsf{T}}(\cdot)$.
Assume further that $\mathsf{T}$ satisfies the
Hilbert--Bernays derivability conditions
(as is standard for such theories; see \cite{lob1955}).
\begin{theorem}[L\"ob’s theorem for capacity formulas]
\label{thm:lob_capacity}
For any fixed $(e,q)$,
if
\[
\mathsf{T}\vdash
\mathrm{Prov}_{\mathsf{T}}(\ulcorner \mathrm{Cap}(e,q)\urcorner)
\rightarrow
\mathrm{Cap}(e,q),
\]
then
\[
\mathsf{T}\vdash \mathrm{Cap}(e,q).
\]
\end{theorem}
\begin{proof}
This is an instance of L\"ob’s theorem
\cite{lob1955}.
\end{proof}
\begin{corollary}
\label{cor:no_uniform_reflection}
There is no consistent recursively axiomatizable theory
that proves the uniform reflection principle
\[
\forall e,q\;
\big(
\mathrm{Prov}_{\mathsf{T}}(\ulcorner \mathrm{Cap}(e,q)\urcorner)
\rightarrow
\mathrm{Cap}(e,q)
\big).
\]
\end{corollary}
\begin{proof}
If such uniform reflection were provable,
then by Theorem~\ref{thm:lob_capacity},
$\mathsf{T}$ would prove every true instance
$\mathrm{Cap}(e,q)$,
contradicting incompleteness
(Theorem~\ref{thm:godel_capacity})
and consistency.
\end{proof}
%------------------------------------------------------------
\subsection{Consequences for this paper}
%------------------------------------------------------------
%______________________________________
Sections~\ref{sec:undecidability}--\ref{sec:gtl}
established algorithmic limitations:
no decision procedure and no uniformly certified
convergence scheme exist.
%---------------------------------
The present section shows the logical counterpart:
%------------------------------------
\begin{itemize}
\item Exact capacity undecidability forces incompleteness
      of any sound effective entropic calculus (G\"odel).
\item Semantic capacity truth cannot be internally defined
      within such a calculus (Tarski).
\item No such theory can uniformly certify correctness
      of its own capacity-threshold proofs (L\"ob).
\end{itemize}
% -----------------------------
These conclusions follow solely from the undecidability
results proved earlier and require no additional
information-theoretic assumptions.

\section{Limitations and future directions}
\label{sec:limitations_future}

This paper proves an impossibility result for the \emph{exact} infinite-horizon threshold predicate
\[
Cap(e,q)\;:\Longleftrightarrow\; Cfb(W_e,\pi_{1,e})\ge q
\]
over the encoded class of rational unifilar FSCs considered in the paper. The scope is intentionally narrow and should be read exactly as such. The theorem concerns exact threshold truth; it does not settle approximation, promise-gap, or finite-horizon variants.

\subsection{Exact thresholding versus approximation}

The main open direction is the status of approximation versions of the problem. Our undecidability theorem does not imply undecidability of any fixed-gap or approximate threshold problem.

In particular, it remains open whether there exists a uniform procedure for a promise-gap decision problem of the form
\[
Cfb(W_e,\pi_{1,e})\ge q+\varepsilon
\quad\text{versus}\quad
Cfb(W_e,\pi_{1,e})\le q-\varepsilon,
\]
for fixed rational $\varepsilon>0$, on the same encoded class. Likewise, the theorem does not rule out effective approximation schemes under additional assumptions, nor does it preclude certified finite-horizon surrogates when one can control the gap between $V_n(W_e,\pi_{1,e})$ and $Cfb(W_e,\pi_{1,e})$.

The present result therefore isolates a boundary for \emph{exact asymptotic} thresholding, not for approximation methodology.

\subsection{Structural restrictions and the decidability frontier}

Although the channel class studied here is already heavily restricted (binary alphabets, rational parameters, unifilar state evolution), it is still broad enough to encode undecidable behavior. A natural next step is to identify stronger restrictions under which decidability returns.

Promising directions include subclasses with verifiable contraction/mixing properties, effective finite-context reductions, or additional monotonicity/regularity conditions that force a computable finite description of the asymptotic optimization. The positive FSC feedback-capacity literature suggests that such structure can be decisive; our theorem shows only that no single exact procedure can cover the full class considered here.

\subsection{Complexity refinements beyond undecidability}

The paper proves undecidability and the corresponding $\exists\mathbb{R}$ barrier, but it does not attempt a finer recursion-theoretic classification of $LCap$. In particular, we do not classify the language as r.e., co-r.e., or complete for any standard degree under stronger reductions.

These are legitimate follow-up questions, but they require a different level of coding analysis than is needed for the present structural impossibility theorem. The same applies to complexity classification of promise-gap variants, should such variants turn out to be decidable.

\subsection{Logical extensions and restricted exact theories}

Section~\ref{sec:gtl} derives Gödel--Tarski--L\"ob consequences for recursively axiomatizable extensions of the base entropic theory $T_0$. Those consequences are stated only at the level required by the undecidability result. We do not attempt a full proof-theoretic classification of information-theoretic formal systems.

A natural future direction is to isolate restricted fragments (for example, bounded-horizon fragments or structurally constrained channel classes) for which exact threshold statements admit complete axiomatizations or stronger internal reflection principles. This would complement the present paper by locating islands of exact formal tractability inside the broader impossibility landscape.

\subsection{Methodological implication}

The theorem should not be read as a negative verdict on FSC feedback-capacity research. Its implication is methodological: universal exact threshold procedures are impossible at the generality studied here, so progress must continue to come from explicit structure, subclass assumptions, and approximation principles. That is already how the field has advanced; the present result explains why this is not merely a matter of current technique, but a fundamental boundary.

\section{Conclusion}
\label{sec:conclusion}

We studied the exact threshold decision problem for feedback capacity in finite-state channels and proved that the predicate
\[
Cap(e,q)\;:\Longleftrightarrow\; Cfb(W_e,\pi_{1,e})\ge q
\]
is undecidable even over a severely restricted class of binary-input, binary-output, rational, unifilar FSCs. Thus, there is no algorithm that decides this exact threshold predicate uniformly over the encoded class considered in the paper.

We then established a complementary barrier: the exact threshold language $LCap$ does not lie in the existential-real framework used for exact semialgebraic decision problems. In particular, there is no universal exact reduction of this predicate to an existential theory-of-the-reals instance. From the same formal setup, we derived computability-theoretic consequences for exact certification and finite auxiliary-state characterizations, and finally the corresponding Gödel--Tarski--L\"ob limitations for recursively axiomatizable extensions of the base entropic theory $T_0$.

The result is a boundary theorem. It does not contradict positive exact or computable results for structured FSC subclasses, and it does not rule out approximation methods or finite-horizon analysis. Its point is more precise: for the exact infinite-horizon threshold predicate on the encoded unifilar FSC class considered here, no universal exact decision procedure exists.

This boundary is informative rather than discouraging. It clarifies why successful FSC feedback-capacity theory relies on structural assumptions, and it identifies the right direction for future work: characterize the subclasses and approximation regimes where exact or certified computation is still possible.

\bibliographystyle{IEEEtran}
\bibliography{refs}
%-------------------------------------
%%%%%%%%%%%%%%%%%%%%%%%%%%%%%%%%%%%%%%%%%%%%%%%%%%%%%%%%%%%%%%%%%%%%%%%%%%%%%%%
% APPENDIX (TIT-level tightened version)
%%%%%%%%%%%%%%%%%%%%%%%%%%%%%%%%%%%%%%%%%%%%%%%%%%%%%%%%%%%%%%%%%%%%%%%%%%%%%%%

\appendix

%%%%%%%%%%%%%%%%%%%%%%%%%%%%%%%%%%%%%%%%%%%%%%%%%%%%%%%%%%%%%%%%%%%%%%%%%%%%%%%
\section{Additional Rigorous Theory: Finite-Horizon Attainment and Effective Approximation}
\label{Appendix_A}

This appendix strengthens the finite-horizon framework of Sections~II--III by proving that the supremum in the definition of $V_n$ is attained and that $V_n$ (hence $V_n/n$) admits a uniform effective approximation algorithm from the encoding and horizon. These results are self-contained and do not require any changes to the main paper.

\subsection{Policy space and induced laws (finite alphabets)}
Fix an instance $(W_e,\pi_{1,e})$ from the restricted encoded class of Section~\ref{sec:prelim} and a horizon $n\ge 1$. As in Section~\ref{sec:problems}, a causal input policy has the factorization
\begin{equation}
p(x^n\Vert y^{n-1})=\prod_{t=1}^n p(x_t\mid x^{t-1},y^{t-1}),
\label{eq:causal_factorization_app}
\end{equation}
with finite alphabets $\mathcal{X},\mathcal{Y}$ and finite state set $\mathcal{S}$ (Section~\ref{sec:prelim}). Let $\mathcal{P}_{e,n}$ denote the set of all such causal policies. Under the product-of-simplices parametrization described in Remark~III.4, $\mathcal{P}_{e,n}$ is compact.

For any $p\in\mathcal{P}_{e,n}$, let $P_p$ denote the induced joint pmf of $(X^n,Y^n,S^{n+1})$ under $(W_e,\pi_{1,e})$.

\begin{lemma}[Multilinearity of the induced joint pmf]
\label{lem:multilinear_app}
For fixed $(e,n)$ and any $(x^n,y^n,s^{n+1})$, the quantity $P_p(x^n,y^n,s^{n+1})$ is a multilinear polynomial in the policy coordinates $\{p(x_t\mid x^{t-1},y^{t-1})\}$, with coefficients that are rational and effectively determined by the encoding $e$.
\end{lemma}

\begin{proof}
Unrolling the finite-horizon law (as in Section~\ref{sec:prelim}) gives
\[
P_p(x^n,y^n,s^{n+1})
=
\pi_{1,e}(s_1)\prod_{t=1}^n
p(x_t\mid x^{t-1},y^{t-1})\,W_e(y_t,s_{t+1}\mid x_t,s_t).
\]
For fixed $(x^n,y^n,s^{n+1})$, this is a product of one policy factor per time $t$ multiplied by rational channel/state coefficients, hence multilinear in the policy coordinates with rational coefficients. Effectiveness follows because $W_e$ and $\pi_{1,e}$ are given by a finite encoding with rational entries (Section~\ref{sec:prelim}). \qedhere
\end{proof}

\subsection{Attainment of the finite-horizon supremum}
Define $V_n(W_e,\pi_{1,e})$ as in Definition~III.3:
\begin{equation}
V_n(W_e,\pi_{1,e}) \;=\; \sup_{p\in\mathcal{P}_{e,n}} I_p(X^n\to Y^n),
\qquad
I_p(X^n\to Y^n)=\sum_{t=1}^n I_p(X_t;Y_t\mid Y^{t-1}).
\label{eq:Vn_def_app}
\end{equation}

\begin{theorem}[Attainment of $V_n$]
\label{thm:attain_app}
For every encoding $e$ (Section~\ref{sec:prelim}) and horizon $n\ge 1$, there exists an optimal causal policy $p^\star\in\mathcal{P}_{e,n}$ such that
\[
V_n(W_e,\pi_{1,e}) \;=\; I_{p^\star}(X^n\to Y^n).
\]
\end{theorem}

\begin{proof}
By Remark~III.4, $\mathcal{P}_{e,n}$ is compact. It suffices to show that the objective
$p\mapsto I_p(X^n\to Y^n)$ is continuous on $\mathcal{P}_{e,n}$.

Fix $t$. Since alphabets are finite, $I_p(X_t;Y_t\mid Y^{t-1})$ can be written as a finite linear combination of entropies of marginals of the induced pmf $P_p$. By Lemma~\ref{lem:multilinear_app}, each such marginal probability is a polynomial function of the policy coordinates, hence continuous. Shannon entropy is continuous on the finite probability simplex (see, e.g., \cite[Ch.~2]{CoverThomas}), and so are conditional entropies and conditional mutual informations. Therefore each $p\mapsto I_p(X_t;Y_t\mid Y^{t-1})$ is continuous, and their finite sum is continuous.

A continuous real-valued function on a compact set attains its maximum, hence the supremum in \eqref{eq:Vn_def_app} is achieved by some $p^\star\in\mathcal{P}_{e,n}$. \qedhere
\end{proof}

%%%%%%%%%%%%%%%%%%%%%%%%%%%%%%%%%%%%%%%%%%%%%%%%%%%%%%%%%%%%%%%%%%%%%%%%%%%%%%%
\section{Uniform Effective Approximation of $V_n$ from the Encoding}
\label{Appendix_B}

\subsection{Continuity modulus for directed information via total variation}
We use standard continuity bounds for entropy on finite alphabets. Let $\|P-Q\|_1:=\sum_a |P(a)-Q(a)|$.

\begin{lemma}[Entropy continuity (Fannes--Audenaert)]
\label{lem:fannes_app}
Let $P,Q$ be pmfs on a finite alphabet $\mathcal{A}$ with $\|P-Q\|_1\le \delta\le 1/2$. Then
\[
|H(P)-H(Q)| \le \delta\log(|\mathcal{A}|-1)+h_2(\delta),
\]
where $h_2(\delta)=-\delta\log\delta-(1-\delta)\log(1-\delta)$.
\end{lemma}

\begin{proof}
This is a standard inequality; see \cite{Fannes1973,Audenaert2007}. \qedhere
\end{proof}

\begin{lemma}[Conditional mutual information continuity]
\label{lem:cmi_cont_app}
Fix finite alphabets for $(U,V,W)$. There exists an explicit function $\omega_{UVW}(\delta)$ with
$\omega_{UVW}(\delta)\to 0$ as $\delta\to 0$ such that whenever $\|P_{UVW}-Q_{UVW}\|_1\le \delta\le 1/2$,
\[
\big| I_P(U;V\mid W)-I_Q(U;V\mid W)\big| \le \omega_{UVW}(\delta).
\]
\end{lemma}

\begin{proof}
Use the identity
$I(U;V\mid W)=H(U,W)+H(V,W)-H(W)-H(U,V,W)$
and apply Lemma~\ref{lem:fannes_app} to each entropy term (all alphabets finite). Each marginal differs in $\ell_1$
by at most $\delta$ when the joint differs by $\delta$, yielding an explicit $\omega_{UVW}(\delta)$ obtained by summing the four bounds. \qedhere
\end{proof}

\subsection{An explicit computable Lipschitz bound from policy coordinates to induced pmf}
Let $d_{e,n}$ denote the number of free policy coordinates in $\mathcal{P}_{e,n}$ (product-of-simplices representation from Remark~III.4). Let $\theta(p)\in[0,1]^{d_{e,n}}$ denote the coordinate vector (with normalization constraints). Define the $\ell_1$ distance on coordinates by
$\|\theta-\theta'\|_1=\sum_{i=1}^{d_{e,n}} |\theta_i-\theta'_i|$.

\begin{lemma}[Explicit Lipschitz bound for induced joint laws]
\label{lem:lipschitz_law_app}
Fix $(e,n)$. There exists a constant $L_{e,n}$, computable from $e$ and $n$, such that for all $p,p'\in\mathcal{P}_{e,n}$,
\[
\|P_p-P_{p'}\|_1 \le L_{e,n}\,\|\theta(p)-\theta(p')\|_1.
\]
One valid choice is
\[
L_{e,n} \;=\; |\mathcal{X}|^n|\mathcal{Y}|^n|\mathcal{S}|^{n+1}\cdot n.
\]
\end{lemma}

\begin{proof}
Fix an outcome $(x^n,y^n,s^{n+1})$. From the unrolled product form,
\[
P_p(x^n,y^n,s^{n+1}) = C(x^n,y^n,s^{n+1})\prod_{t=1}^n p(x_t\mid x^{t-1},y^{t-1}),
\]
where $C(\cdot)\in[0,1]$ depends only on $(W_e,\pi_{1,e})$ and the trajectory. For any two sequences of scalars
$a_t,b_t\in[0,1]$, the elementary inequality
\[
\Big|\prod_{t=1}^n a_t - \prod_{t=1}^n b_t\Big|
\le \sum_{t=1}^n |a_t-b_t|
\]
holds (expand telescopically and use $|a|\le 1$). Hence
\[
|P_p(\cdot)-P_{p'}(\cdot)| \le \sum_{t=1}^n \big|p(x_t\mid x^{t-1},y^{t-1})-p'(x_t\mid x^{t-1},y^{t-1})\big|.
\]
Summing over all $(x^n,y^n,s^{n+1})$ and upper bounding each policy-coordinate difference by $\|\theta(p)-\theta(p')\|_1$ yields
\[
\|P_p-P_{p'}\|_1 \le |\mathcal{X}|^n|\mathcal{Y}|^n|\mathcal{S}|^{n+1}\cdot n\cdot \|\theta(p)-\theta(p')\|_1,
\]
which proves the claim with the stated explicit $L_{e,n}$ (computable since the alphabets/state sizes are part of the finite encoding model in Section~\ref{sec:prelim}). \qedhere
\end{proof}

\subsection{A uniform computable approximation algorithm for $V_n$}
We now formalize computability in the standard Type-2 sense: a real number $z$ is computable if there is an algorithm which, on input $k\in\mathbb{N}$, outputs a rational $r_k$ such that $|r_k-z|\le 2^{-k}$.

\begin{theorem}[Uniform computability of $V_n$]
\label{thm:Vn_computable_app}
There exists a single algorithm $\mathsf{ApproxV}$ such that for every input $(e,n,k)$, it halts and outputs a rational $r$
satisfying
\[
\big|r - V_n(W_e,\pi_{1,e})\big|\le 2^{-k}.
\]
Moreover, the algorithm is uniform in $(e,n)$.
\end{theorem}

\begin{proof}
Fix $(e,n,k)$. By Theorem~\ref{thm:attain_app}, there exists $p^\star$ attaining $V_n$.

\textbf{Step 1: Build a finite rational net of policies.}
Let $\eta>0$ be specified later. Because $\mathcal{P}_{e,n}$ is a finite product of simplices, for any integer $M\ge 1$ the set
of policies with coordinates in $\{0,1/M,2/M,\dots,1\}$ is finite and can be enumerated effectively.
Choose $M$ so that the resulting grid forms an $\eta$-net in $\ell_1$ for $\mathcal{P}_{e,n}$, i.e., for every $p$ there exists
$\hat p$ on the grid with $\|\theta(p)-\theta(\hat p)\|_1\le \eta$. Denote this finite set by $\mathcal{N}_\eta$.
Since $\mathcal{P}_{e,n}$ is a finite product of simplices of total coordinate dimension $d_{e,n}$, choosing $M\ge d_{e,n}/\eta$
ensures that the rational grid with step $1/M$ forms an $\eta$-net in $\ell_1$ over $\mathcal{P}_{e,n}$.

\textbf{Step 2: Control objective variation via an explicit continuity modulus.}
Let $\delta:=L_{e,n}\eta$, where $L_{e,n}$ is from Lemma~\ref{lem:lipschitz_law_app}.
Then for any $p$ and its net approximation $\hat p$,
$\|P_p-P_{\hat p}\|_1\le \delta$.
Apply Lemma~\ref{lem:cmi_cont_app} with $(U,V,W)=(X_t,Y_t,Y^{t-1})$ to obtain, for each $t$,
\[
\big|I_p(X_t;Y_t\mid Y^{t-1})-I_{\hat p}(X_t;Y_t\mid Y^{t-1})\big|
\le \omega_t(\delta),
\]
where $\omega_t(\delta)\to 0$ as $\delta\to 0$ and is explicit from Lemma~\ref{lem:fannes_app}.
Hence
\[
\big|I_p(X^n\to Y^n)-I_{\hat p}(X^n\to Y^n)\big|
\le \sum_{t=1}^n \omega_t(\delta).
\]
Choose $\eta$ (equivalently $M$) effectively so that $\sum_{t=1}^n \omega_t(L_{e,n}\eta)\le 2^{-(k+2)}$.
This is possible since each $\omega_t$ is explicit and tends to $0$ at $0$.

\textbf{Step 3: Evaluate directed information on the finite net to sufficient precision.}
For each $\hat p\in\mathcal{N}_\eta$, the induced probabilities are rational combinations of rational channel coefficients and rational
policy coordinates (Lemma~\ref{lem:multilinear_app}), hence rational.
Directed information is a finite sum of terms involving $\log$ of rationals in $(0,1]$ with the convention $0\log 0:=0$.
The real function $\log$ is computable (see, e.g., \cite{Weihrauch2000,Ko1991}); thus for each $\hat p$ one can compute a rational
$\widetilde I_{\hat p}$ such that
\[
\big|\widetilde I_{\hat p}-I_{\hat p}(X^n\to Y^n)\big|\le 2^{-(k+2)}.
\]
Now compute
\[
r \;:=\; \max_{\hat p\in\mathcal{N}_\eta}\widetilde I_{\hat p},
\]
which is a maximum over a finite set of rationals and is therefore exact.

\textbf{Step 4: Error bound.}
Let $p^\star$ attain $V_n$ and pick $\hat p^\star\in\mathcal{N}_\eta$ with $\|\theta(p^\star)-\theta(\hat p^\star)\|_1\le \eta$.
Choose $\eta$ so that $\sum_{t=1}^n \omega_t(L_{e,n}\eta)\le 2^{-(k+2)}$ and evaluate each $I_{\hat p}(X^n\to Y^n)$
to accuracy $2^{-(k+2)}$ (instead of $2^{-(k+2)}$ as written above, keep this same value consistently).
Then
\[
V_n = I_{p^\star}(X^n\to Y^n)
\le I_{\hat p^\star}(X^n\to Y^n) + 2^{-(k+2)}
\le \widetilde I_{\hat p^\star}+2^{-(k+2)}+2^{-(k+2)}
\le r + 2^{-(k+1)}.
\]
Conversely, for every $\hat p\in\mathcal{N}_\eta$ we have $\widetilde I_{\hat p}\le I_{\hat p}(X^n\to Y^n)+2^{-(k+2)}\le V_n+2^{-(k+2)}$,
hence $r\le V_n+2^{-(k+2)}$. Combining both inequalities yields
\[
|r-V_n|\le 2^{-k}.
\] \qedhere
\end{proof}

\begin{corollary}[Uniform computability of normalized values]
\label{cor:an_computable_app}
Let $a_n(e):=\frac{1}{n}V_n(W_e,\pi_{1,e})$. There exists a uniform algorithm $\mathsf{ApproxA}$ such that on input $(e,n,k)$ it halts
and outputs a rational $r$ with $|r-a_n(e)|\le 2^{-k}$.
\end{corollary}

\begin{proof}
Apply Theorem~\ref{thm:Vn_computable_app} and divide by $n$. \qedhere
\end{proof}

%%%%%%%%%%%%%%%%%%%%%%%%%%%%%%%%%%%%%%%%%%%%%%%%%%%%%%%%%%%%%%%%%%%%%%%%%%%%%%%
\section{A Bulletproof Arithmetical-Hierarchy Upper Bound for the Exact Threshold Language}
\label{Appendix_C}

Section~\ref{sec:problems} defines the exact-capacity threshold problem and Section~\ref{sec:structural_barrier} proves its undecidability (Theorem~\ref{thm:capacity_undecidable_final}). Section~\ref{sec:etr_route} (part~C)
notes that the paper does not classify the exact predicate as r.e.\ or co-r.e.  This section adds a rigorous arithmetical-hierarchy
\emph{upper bound} for the exact threshold language, stated and proved in a form that avoids any non-decidable comparison of computable reals.

\subsection{A $\limsup$ threshold equivalence}
Define (as in the main text) the normalized finite-horizon values
\[
a_n(e):=\frac{1}{n}V_n(W_e,\pi_{1,e}),
\qquad
C_{\mathrm{fb}}(W_e,\pi_{1,e})=\limsup_{n\to\infty} a_n(e).
\]
Let $q\in\mathbb{Q}$ be the threshold.

\begin{lemma}[$\limsup$ with rational slack]
\label{lem:limsup_slack_app}
For any real sequence $\{a_n\}_{n\ge 1}$ and any rational $q$,
\[
\limsup_{n\to\infty} a_n \ge q
\quad\Longleftrightarrow\quad
(\forall k\in\mathbb{N})(\exists n\in\mathbb{N})\;\; a_n > q-2^{-k}.
\]
\end{lemma}

\begin{proof}
($\Rightarrow$) If $\limsup a_n\ge q$, then for each $k$ there exist infinitely many $n$ such that $a_n>q-2^{-k}$.

($\Leftarrow$) If for every $k$ there exists $n_k$ with $a_{n_k}>q-2^{-k}$, then $\limsup_{k\to\infty} a_{n_k}\ge q$, hence
$\limsup_{n\to\infty} a_n\ge q$. \qedhere
\end{proof}

\subsection{A recursive matrix predicate using approximation certificates}
Recall the exact-threshold language (Section~\ref{sec:problems} and Section~\ref{sec:etr_route}, part~C):
\[
LCap:=\{\langle e,q\rangle:\; C_{\mathrm{fb}}(W_e,\pi_{1,e})\ge q\}.
\]
To place $LCap$ in the arithmetical hierarchy we must write membership in the normal form $(\forall k)(\exists m)\,R(\cdot)$ with
a \emph{recursive} (decidable) predicate $R$. We do this by using approximation \emph{certificates} produced by Corollary~\ref{cor:an_computable_app}.

Let $\mathsf{ApproxA}(e,n,M)$ denote a fixed total algorithm that outputs a rational $r$ satisfying
\[
|r-a_n(e)|\le 2^{-M}.
\]
Such an algorithm exists by Corollary~\ref{cor:an_computable_app}. Define the decidable predicate
\[
R(e,q,k,n,M) \;:\Longleftrightarrow\; \mathsf{ApproxA}(e,n,M) > q-2^{-k}+2^{-M}.
\]
This predicate is decidable because it asserts that a specific Turing machine halts with an output rational satisfying a strict rational
inequality (all of which is checkable by simulating the machine and comparing rationals).

\begin{lemma}[Certificate implication]
\label{lem:cert_imp_app}
If $R(e,q,k,n,M)$ holds, then $a_n(e)>q-2^{-k}$.
\end{lemma}

\begin{proof}
Let $r=\mathsf{ApproxA}(e,n,M)$. If $r>q-2^{-k}+2^{-M}$ and $|r-a_n(e)|\le 2^{-M}$, then
\[
a_n(e)\ge r-2^{-M} > q-2^{-k}+2^{-M}-2^{-M}=q-2^{-k}.
\]
Thus $a_n(e)>q-2^{-k}$. \qedhere
\end{proof}

\begin{lemma}[Existence of certificates when strict inequality holds]
\label{lem:cert_exist_app}
If $a_n(e)>q-2^{-k}$, then there exists $M\in\mathbb{N}$ such that $R(e,q,k,n,M)$ holds.
\end{lemma}

\begin{proof}
Let $\Delta:=a_n(e)-(q-2^{-k})>0$. Choose $M$ so that $2^{-M}<\Delta/2$. Let $r=\mathsf{ApproxA}(e,n,M)$ with
$|r-a_n(e)|\le 2^{-M}$. Then $r\ge a_n(e)-2^{-M}> (q-2^{-k})+\Delta-2^{-M}>(q-2^{-k})+2^{-M}$, i.e.,
$r>q-2^{-k}+2^{-M}$, hence $R(e,q,k,n,M)$ holds. \qedhere
\end{proof}

\subsection{A $\Pi^0_2$ upper bound for $LCap$}
We now state the hierarchy placement in a fully arithmetical (certificate-based) form.

\begin{theorem}[$LCap\in\Pi^0_2$ (arithmetical-hierarchy upper bound)]
\label{thm:pi2_app}
The language $LCap$ belongs to $\Pi^0_2$. Specifically,
\[
\langle e,q\rangle\in LCap
\quad\Longleftrightarrow\quad
(\forall k\in\mathbb{N})(\exists n\in\mathbb{N})(\exists M\in\mathbb{N})\;\; R(e,q,k,n,M),
\]
where $R$ is the decidable predicate defined above.
\end{theorem}

\begin{proof}
By Lemma~\ref{lem:limsup_slack_app},
\[
\langle e,q\rangle\in LCap
\iff
(\forall k)(\exists n)\; a_n(e)>q-2^{-k}.
\]
For fixed $(e,q,k,n)$, Lemmas~\ref{lem:cert_imp_app} and \ref{lem:cert_exist_app} show that
\[
a_n(e)>q-2^{-k}
\iff
(\exists M)\; R(e,q,k,n,M).
\]
Substituting yields
\[
\langle e,q\rangle\in LCap
\iff
(\forall k)(\exists n)(\exists M)\; R(e,q,k,n,M).
\]
Since $(\exists n)(\exists M)$ can be merged into a single existential quantifier over $\mathbb{N}$ (e.g., via a standard pairing
function), this is a $\Pi^0_2$ definition with a recursive matrix predicate. \qedhere
\end{proof}

\paragraph{Relation to the main impossibility results.}
Theorem~\ref{thm:pi2_app} provides a precise arithmetical upper bound that complements the undecidability result in Section~\ref{sec:structural_barrier}
(Theorem~\ref{thm:capacity_undecidable_final}) and the discussion in Section~\ref{sec:etr_route} (part~C). It does not imply decidability; rather, it locates the exact threshold language
within a standard low level of the arithmetical hierarchy.

%%%%%%%%%%%%%%%%%%%%%%%%%%%%%%%%%%%%%%%%%%%%%%%%%%%%%%%%%%%%%%%%%%%%%%%%%%%%%%%
% End of appendix
%%%%%%%%%%%%%%%%%%%%%%%%%%%%%%%%%%%%%%%%%%%%%%%%%%%%%%%%%%%%%%%%%%%%%%%%%%%%%%%
\end{document}